\documentclass[sort&compress]{elsarticle}

\usepackage{lineno,hyperref}
\modulolinenumbers[5]

\journal{Journal of Logical and Algebraic Methods in Programming}
\date{November 15, 2021}
\usepackage{ebproof} 
\usepackage{todonotes} 
\usepackage{amsmath} 
\usepackage{amssymb} 
\usepackage{amsthm} 
\usepackage{listings} 
\usepackage{cleveref}
\usepackage{placeins} 

\newtheorem{theorem}{Theorem}[section]
\newtheorem{lemma}[theorem]{Lemma}
\newtheorem{example}[theorem]{Example}
\newtheorem*{theorem*}{Theorem}
\newtheorem{definition}{Definition}[section]
\newtheorem{corollary}[theorem]{Corollary}

\newcommand{\hda}{\mid}

\newcommand{\inl}[1]{\textit{inl}\ #1}
\newcommand{\inr}[1]{\textit{inr}\ #1}
\newcommand\doubleplus{+\kern-1.3ex+\kern0.8ex}
\newcommand{\concat}{\doubleplus}
\newcommand{\sub}{\sigma}

\newcommand{\fun}[4]{\mathtt{fun\ } #1(#2, \dots, #3) \rightarrow #4}
\newcommand{\funz}[2]{\mathtt{fun\ } #1() \rightarrow #2}
\newcommand{\case}[4]{\mathtt{case\ } #1 \mathtt{\ of\ } #2 \mathtt{\ then\ } #3 \mathtt{\ else\ } #4}
\newcommand{\letrec}[5]{\mathtt{rec\ } #1(#2, \dots, #3) \rightarrow #4 \mathtt{\ in\ } #5}

\newcommand{\elet}[3]{\mathtt{let\ } #1 = #2 \mathtt{\ in\ } #3}
\newcommand{\apply}[3]{\mathtt{apply\ } #1(#2, \dots, #3)}
\newcommand{\applyz}[1]{\mathtt{apply\ } #1()}
\newcommand{\idfs}{\mathcal{I}d}
\newcommand{\consfs}[2]{#1 :: #2}

\newcommand{\ctxpre}{\leq_{\textit{ctx}}}
\newcommand{\ctxequiv}{\equiv_{\textit{ctx}}}
\newcommand{\ciupre}{\leq_{\textit{ciu}}}
\newcommand{\ciuequiv}{\equiv_{\textit{ciu}}}
\newcommand{\behpre}{\leq_{b}}
\newcommand{\behequiv}{\equiv_{b}}

\newcommand{\expscoped}[2]{#1 \vdash_{e} #2}
\newcommand{\valscoped}[2]{#1 \vdash_{v} #2}
\newcommand{\subscoped}[3]{#1 \vdash_{s} #2 :: #3}
\newcommand{\framesclosed}[1]{\mathtt{FSC\ } #1}
\newcommand{\frameclosed}[1]{\mathtt{FC\ } #1}

\newcommand{\subst}[2]{#1[#2]}
\newcommand{\idsubst}{\textit{id}}
\newcommand{\restrict}{\setminus}
\newcommand{\preserves}[2]{preserves(#1, #2)}

\newcommand{\termk}[3]{\langle #1, #2 \rangle \Downarrow^{#3}}
\newcommand{\term}[2]{\langle #1, #2 \rangle \Downarrow}
\newcommand{\rewrites}[4]{\langle #1, #2 \rangle \longrightarrow \langle #3 , #4 \rangle}
\newcommand{\rewritesbr}[4]{\langle #1, #2 \rangle \longrightarrow \\ &\qquad\langle #3 , #4 \rangle}
\newcommand{\rewritesn}[5]{\langle #1, #2 \rangle \longrightarrow^{#3} \langle #4 , #5 \rangle}
\newcommand{\rewritesstar}[3]{\langle #1, #2 \rangle \longrightarrow^* #3}

\begin{document}

\begin{frontmatter}

\title{Program Equivalence in an Untyped, Call-by-value Lambda Calculus with Uncurried Recursive Functions}

\author[elte]{Dániel Horpácsi\corref{corr}}
\ead{daniel-h@elte.hu}

\author[elte]{Péter Bereczky\corref{corr}}
\ead{berpeti@inf.elte.hu}
\cortext[corr]{Corresponding author}

\author[elte,kent]{Simon Thompson}
\ead{s.j.thompson@kent.ac.uk}

\address[elte]{ELTE, Eötvös Loránd University, Hungary}
\address[kent]{University of Kent, United Kingdom}




\begin{abstract}

We aim to reason about the correctness of behaviour-preserving transformations of Erlang programs. Behaviour preservation is characterised by semantic equivalence. Based upon our existing formal semantics for Core Erlang, we investigate potential definitions of suitable equivalence relations. In particular we adapt a number of existing approaches of expression equivalence to a simple functional programming language that carries the main features of sequential Core Erlang; we then examine the properties of the equivalence relations and formally establish connections between them.
The results presented in this paper, including all theorems and their proofs, have been machine checked using the Coq proof assistant.
\end{abstract}


\end{frontmatter}


\section{Introduction}

Most language processors and refactoring tools lack a precise formal specification of how the code is affected by the changes they may make. In particular, refactoring tools are expected not to change the behaviour of any program, but in practice, this property is only validated by testing. This form of verification may or may not provide trust in users willing to refactor industrial-scale code. Higher assurance can be achieved by making formal arguments to verify behaviour preservation. This requires a rigorous formal definition of the programs being refactored, a precise description of the effect of refactorings on these programs, and a suitable definition of program equivalence.

\subsection{Refactoring and Program Equivalence}

A program transformation is said to be a refactoring if it preserves the observable behaviour~\cite{fowler1999refactoring}. The verification of whether a transformation is refactoring can be based upon the concept of semantic program equivalence: the internal workings of the software may be altered, but the result has to be observationally indistinguishable from the original in any arbitrary environment.


The typical refactoring process -- reworking a piece of code to increase its quality -- can be decomposed into multiple, smaller refactoring steps (known as \emph{micro-refactorings}). These smaller steps may vary from eliminating unused variables to extracting code portions to function abstractions. Some are local to particular portions of code -- such as an expression, a function or a module -- but some, such as renaming a global function, make extensive modifications across a code base, which seems to require the entire software project to be examined for equivalence when reasoning about the correctness of a complex refactoring.

Syntactically local refactorings are relatively straightforward to define and verify: the expected code changes can be specified with conditional term rewrite rules, and their correctness can be checked via contextual equivalence with side conditions. 
Many refactorings are not local, however, and they span across compilation units via semantic dependencies such as name binding, data flow or control flow. For example, renaming of a function will affect not only the module in which the function is defined, but also every module in which the function is applied. Essentially, extensive refactoring is a special composition of local changes (expressed with conditional term rewrite rules), connected by semantic dependencies in the program. 



A key observation in our previous work~\cite{horpacsi2017trustworthy,horpacsi2016towards} is that, in the case of extensive transformations, it is possible to separate local equivalence from global equivalence. The former property ensures consistency among local changes, while the latter considers the locally consistent changes in arbitrary global context. With this, 
verification of extensive steps separates reasoning about the locality to a program slice and the local correctness within the slice (expressed with conditional contextual equivalence), from the argument about how local consistency implies program-wide equivalence.

We formalise this separation in the notion of a \emph{refactoring scheme}: a pre-verified, general transformation strategy used for defining various refactoring steps. Using such a scheme, the verification of each specific refactoring step only requires the formal proof of local correctness. This implies that the verification of both structurally local and 
extensive code transformations is reduced to checking conditional equivalence between program terms. 
The work presented here aims to provide a basis for reasoning about refactorings by establishing the appropriate semantics and definitions of program equivalence for languages like Core Erlang.




\subsection{Scope and Contributions}

Although our ultimate goal is to prove Erlang refactorings correct, as a stepping stone we focus on Erlang's intermediate language (Core Erlang) and we investigate refactoring correctness and program equivalence on this lower-level language first. For readers  unfamiliar with Erlang, we note that Core Erlang contains all the essential features of the full language, so that the work presented here can be extended to the full language in a routine way. 
Core Erlang is also the intermediate language for a number of other programming languages, including Elixir and LFE, meaning that program equivalence in the core language can be used for checking equivalence in these other languages, assuming the existence of a trusted translation from the source language to Core Erlang.


We have already developed natural and functional big-step semantics for sequential Core Erlang~\cite{bereczky2020core,bereczky2020machine,bereczky2021validation} in the Coq proof assistant and defined simple behavioural equivalence. In this paper we take this work a step further by distilling sequential Core Erlang into an extended lambda calculus and systematically defining and comparing various equivalence relations on that system. In contrast with related work with similar goals, the novelty of our approach lies with the Erlang-like features (such as dynamic typing, call-by-value evaluation, pattern matching, and uncurried functions) of the programming language and full, machine-checked formalisation.


The main contributions of the paper are:
\begin{itemize}
	\item A formal semantics for the language under investigation, which is essentially an extended, untyped, strict lambda calculus with uncurried functions.
	\item A formalisation of a number of termination-based equivalence concepts, and a proof of their coincidence for this language. 
	\item A characterisation of behavioural equivalence that is not solely based on termination, and which is also proved to coincide with the other equivalences.
	\item A machine-checked formalisation of the language, definitions and proofs in the Coq proof assistant.
\end{itemize}

The rest of the paper is structured as follows. In \Cref{sec:related} we overview related work on program equivalence for similar languages and point to some influential results that we reused in our work. In \Cref{sec:semantics} we formally define the syntax and semantics of the language we investigate, then in \Cref{sec:naive} we specify various equivalence definitions for Core Erlang, including simple behavioural equivalence~\cite{pierce2010software}. Next in \Cref{sec:equiv} we make these definitions precise using an approach based on logical relations~\cite{pitts2000operational,wand2018contextual}, pointing out their advantages and disadvantages. Finally, \Cref{sec:conclusion} concludes.

\section{Related Work}
\label{sec:related}

In the early stages of our project, we  developed an inductive big-step semantics for sequential Core Erlang~\cite{bereczky2020core,bereczky2020machine} considering exceptions and simple side effects, which has also been implemented in Coq. This semantics is based on related research on Erlang and Core Erlang, e.g. reversible semantics and debugging~\cite{lanese2018theory,lanese2018cauder,nishida2016reversible}, a framework for reasoning about Erlang~\cite{fredlund2001framework} and symbolic execution~\cite{vidal2014towards}. We have also investigated different big-step definition styles~\cite{bereczky2020comparison}, and recently we also implemented an equivalent functional big-step semantics~\cite{owens2016functional} which enabled extensive validation~\cite{bereczky2021validation}.

However, for proving refactoring correctness, we need precise definitions of program equivalence. We have proved some simple expressions equivalent in our formalisation, however, the used equivalence concept is rather an ad hoc definition of behavioural equivalence, which motivates the investigation of (other) precise approaches.


The literature discusses a number of ways to describe and investigate equivalence between programs of a language after having its semantics formally defined. In this section, we provide a brief overview of related results.

The simplest notion of program equivalence is behavioural equivalence~\cite{pierce2010software}, which requires syntactical equality of the program evaluation results. This notion is a rather strict characterisation of the equivalence concept, but may prove sufficient in some simple cases. Clearly, its main advantage is the simple definition lacking any reference to expression contexts.

Another fundamental equivalence definition is the \emph{contextual} (also called \emph{observational}) equivalence, which characterises the congruence property of the equivalence relation. Two programs are contextually equivalent if they are indistinguishable in arbitrary expression contexts observing their behaviour
\footnote{In most cases, the observed behaviour is simply the termination property of the programs, as it is sufficient to ensure the result values to be equivalent by a suitable value equivalence relation. Pitts provides a draft proof of this in~\cite{pitts1997operationally}, but we give a formal proof in \cref{sec:termEquiv}}.
While it is straightforward to disprove programs equivalent with this notion by giving a counterexample context, establishing contextual equivalence requires induction on all expression contexts. To avoid this burdensome induction, related work proposes alternative definitions of the relation, such as bisimulations, CIU (``closed instances of uses'') equivalence and logical relations.


Bisimulations (or applicative bisimulations~\cite{abramsky1990thelazy}) are binary relations between expressions based on their reduction being in the same relation. 
Bisimulation approaches~\cite{pitts1997operationally,simpson2019behavioural} are naturally proved to be sound for contextual equivalence, ensuring that bisimilar programs are equivalent (using Howe's method for proving congruence~\cite{howe1996proving}); therefore, bisimulations can be used to prove expressions equivalent. However, completeness does not necessarily hold as bisimulations sometimes provide a finer approach (defining a stricter equivalence), distinguishing some terms that are contextually equivalent~\cite{dallago2014coinductive}.


The other widespread approach to establishing contextual equivalence is using alternative equivalences (CIU equivalence, logical relations) and proving that these relations coincide with contextual equivalence. The idea of the CIU equivalence originates from Mason and Talcott~\cite{mason1991equivalence}, but since then other authors investigated and used this characterisation with success for imperative languages~\cite{gordon1999compilation}, for variants of lambda calculus~\cite{craig2018triangulating,culpepper2017contextual,wand2018contextual}, lambda calculus with recursive types~\cite{ahmed2006stepindexed,birkedal2013stepindexed}, and with quantified types~\cite{ahmed2006stepindexed}. Some authors also included effects in this characterisation~\cite{mason1991equivalence,benton1999modans}. Most of these results used the same idea: they defined a form of continuation-style semantics (such as reduction semantics, frame-stack semantics) or termination relation, then defined CIU and contextual equivalence and proved that they coincide. The novelty of the various results lie in the choice of the type system and the language constructs under investigation.

Other authors~\cite{pitts2000operational,ahmed2006stepindexed,wand2018contextual,pitts2010step,culpepper2017contextual} also defined logical relations and proved that they also coincide with the previous equivalences. In case of the logical relations, there was two main approach: type-indexed~\cite{pitts2000operational} or step- (and type)-indexed~\cite{ahmed2006stepindexed,wand2018contextual,pitts2010step,culpepper2017contextual}.

There are further approaches to prove program equivalence, for example using algorithms; however, for these to work we need either an operational semantics based on term rewriting~\cite{lucanu2013program}, or reasoning in some dedicated logics like matching logic~\cite{rosu2014language}.

In the remainder of this paper, we investigate step-indexed logical relations, CIU equivalence, and contextual equivalence for a simple untyped functional language  based on the previous results of the above-mentioned authors. We remind the reader that our definitions and theorems are fully formalised in the Coq proof assistant~\cite{coreerlangmini} (similarly to the one by Wand et al. in Coq~\cite{wand2018contextual}, or the one by McLaughlin et al. in Agda~\cite{craig2018triangulating}).

\section{An Untyped, Strict, Uncurried Functional Language with Recursion and Pattern Matching}
\label{sec:semantics}


In this section, we present the syntax and semantics of a functional language that resembles a simple subset of Core Erlang. In later sections we will define equivalence relations over expressions of this language.

\subsection{Syntax}
\label{sec:syntax}

The language we investigate has integers, lists and functions. In syntax, it supports literals and patterns for integers and lists, uncurried function abstraction and application, pattern matching with \verb|case|, \verb|let|-binding and explicitly recursive function abstraction with \verb|rec|. Values are defined as a subset of expressions, and include function closures for representing values of function abstractions. Integer addition is built-in to enable writing meaningful programs in the language. Note that unlike the majority of the cited related work, we also include pattern matching, which allows us observe the differences between values and implement conditionals in the object theory\footnote{Without pattern matching (or similar language features that allow for inspecting and comparing the contents of values) we cannot prove that contextual equivalence (which coincides with CIU equivalence) and behavioural equivalence coincide (\Cref{thm:behciu}).}. 


We use the following three auxiliary functions for pattern matching (we omit the formal definitions):
\begin{itemize}
	\item $\textit{vars}(p)$ is the list of variables used in $p$.
    \item $\textit{is\_match}(p, v)$ decides whether the value $v$ matches the pattern $p$.
    \item $\textit{match}(p, v)$ yields the result of matching $v$ with $p$, as a substitution (\Cref{sec:subst}).
\end{itemize}

For simplicity, in the machine-checked formalisation~\cite{coreerlangmini} we used a nameless variable representation, that is variables and function identifiers are de Bruijn indices~\cite{deBruijn}. This way, in the body of a binder, the outermost indices denote the bound variables, e.g. for functions, the outermost index point to the identifier of the function, while the next $k$ indices denote the formal parameters. One could write non-recursive functions by simply omitting the use of this index. On the other hand, in this paper, we present our results with explicit names for readability, however, we regard alpha equivalent expressions, patterns, values as equal.

\begin{definition}[Syntax of the language]
\begin{flalign*}
p ::=&\ l \hda x \hda [ p_1 | p_2 ] \hda [] \\
v ::=&\ l \hda x \hda f/k \hda \fun{f/k}{x_1}{x_k}{e} \hda [ v_1 | v_2 ] \hda [] \\
e ::=&\ v \hda [ e_1 |e_2 ] \hda \apply{e}{e_1}{e_k} \hda \case{e_1}{p}{e_2}{e_3}  \\
\hda\ &\letrec{f/k}{x_1}{x_k}{e_0}{e}
\hda \elet{x}{e_1}{e_2} \hda e_1 + e_2 
\end{flalign*}
\end{definition}

It is also to be highlighted that, especially in contrast to the work of Wand et al.~\cite{wand2018contextual}, our language does not require parts of compound expressions to be in normal form (i.e. reduced to value by explicit sequencing with \verb|let| expressions). This poses some challenges in the formalisation, but it keeps the gap between the formalised language and Core Erlang reasonable.


\subsection{Substitution}
\label{sec:subst}

Before giving semantics expressions, we define parallel substitutions~\cite{autosubst} as functions mapping names (either variable names or function identifiers) to the union of expressions and names (Coq's notations are used to denote the constructors ($\textit{inl}$, $\textit{inr}$) of the union type). Names are contained within expressions, but we need to include name-to-name mappings explicitly, since conversion from names to expressions is ambiguous (context-dependent) in the de Bruijn representation. 
With this, we can define the identity substitution and substitution update as follows.

\begin{definition}[Identity substitution]
	$\idsubst(x) = \inr{x}$
\end{definition}

\begin{definition}[Substitution update]\label{def:substupdate}
	$(\sub[x \mapsto e])(y) := \begin{cases}
									\inl{e} \text{ if x = y}\\
									\sub(y) \text{ otherwise}	                             
	                             \end{cases}$
\end{definition}

We use $\sub[x_1 \mapsto e_1, \dots, x_k \mapsto e_k]$ to denote the composition of a finite number of updates (it is assumed that $x_1, \dots, x_k$ are all different). We remind the reader that in the de Bruijn representation, a parallel substitution is specified with the replacements $e_1, \dots, e_k$, meaning the indexes $\#1, \dots, \#k$ substituted for the expressions $e_1, \dots, e_k$. We use the notation $\sub \restrict \{x_1, \dots, x_k\}$ to remove the bindings of $x_1, \dots, x_k$ from $\sub$ (i.e. replace the bindings with identity substitution).

Substitutions can be applied to any expressions. We use the common notation $\subst{e}{\sub}$ to denote the application of the substitution $\sub$ to the expression $e$ (replacing the free occurrences of names in $e$ with the corresponding expressions in $\sub$).

\begin{definition}[Application of parallel substitution]
	\begin{flalign*}
		\subst{l}{\sub} :=& l\\
		\subst{[]}{\sub} :=& [] \\
		\subst{x}{\sub} :=& \begin{cases} y \text{ if $\sub(x) = \inr{y}$ ($y$ is a name)} \\ e \text{ if $\sub(x) = \inl{e}$} \end{cases}\\
		\subst{f/k}{\sub} :=& \begin{cases} y \text{ if $\sub(f/k) = \inr{y}$ ($y$ is a name)} \\ e \text{ if $\sub(f/k) = \inl{e}$} \end{cases}\\
		\subst{(\fun{f/k}{x_1}{x_k}{e})}{\sub} :=& \\\fun{f/k}{x_1}{x_k}{(\subst{&e}{\sub \restrict \{f/k, x_1, \dots, x_k\}})}\\
		\subst{(\apply{e}{e_1}{e_k})}{\sub} :=& \apply{\subst{e}{\sub}}{\subst{e_1}{\sub}}{\subst{e_k}{\sub}} \\
		\subst{(\elet{x}{e_1}{e_2})}{\sub} :=& \elet{x}{\subst{e_1}{\sub}}{\subst{e_2}{\sub \restrict \{x\}}}\\
		\subst{(\letrec{f/k}{x_1}{x_k}{e_0}{e})}{\sub} :=&
		\\\letrec{f/k}{x_1}{x_k}{\subst{e_0}{\sub& \restrict \{f/k, x_1, \dots, x_k\}}}{\subst{e}{\sub \restrict \{f/k\}}}\\
		\subst{(e_1 + e_2)}{\sub} :=& \subst{e_1}{\sub} + \subst{e_2}{\sub}\\
		\subst{[e_1 | e_2]}{\sub} :=& [\subst{e_1}{\sub} | \subst{e_2}{\sub}]\\
		\subst{(\case{e_1}{p}{e_2}{e_3})}{\sub} :=& \\\case{\subst{e_1}{\sub}}{&p}{\subst{e_2}{\sub \restrict \textit{vars(p)}}}{\subst{e_3}{\sub}}
	\end{flalign*}
\end{definition}

For simplicity and readability, we handle $\sub(x)$ ($\sub(f/k)$, respectively) as expressions. In the Coq implementation we followed the techniques known from the \emph{autosubst} library~\cite{autosubst} with de Bruijn representation to implement capture-avoidance for substitutions. For more details we refer to the code~\cite{coreerlangmini}.

\subsection{Variable scoping}
\label{sec:scoping}

Note that this small language is untyped, we cannot define a typing judgement and index equivalence relations with types; instead, we define a variable scoping judgement, following the footsteps of Wand et al.~\cite{wand2018contextual}. This allows us to define our equivalence concepts in a scope-indexed way, where expressions can only be equivalent if they have the same scope. The scoping rules are given in \Cref{fig:scoping}.

\begin{definition}[Name scoping of values and expressions]\ \\
	\begin{itemize}
    	\item $\valscoped{\Gamma}{v}$: all free variables of $v$ are in $\Gamma$. $v$ is a closed value, if $\valscoped{\emptyset}{v}$.
   		\item $\expscoped{\Gamma}{e}$: all free variables of $e$ are in $\Gamma$. $e$ is a closed expression, if $\expscoped{\emptyset}{e}$.
	\end{itemize}
\end{definition}

\begin{figure}[htb]
    \centering
    
    \begin{prooftree}
    \infer0{\valscoped{\Gamma}{l}}
    \end{prooftree}
    \hfill
    \begin{prooftree}
    \hypo{x \in \Gamma}
    \infer1{\valscoped{\Gamma}{x}}
    \end{prooftree}
    \hfill
    \begin{prooftree}
    \hypo{f/k \in \Gamma}
    \infer1{\valscoped{\Gamma}{f/k}}
    \end{prooftree}
    \hfill
    \begin{prooftree}
    \hypo{\expscoped{\Gamma \cup \{f/k, x_1, \dots, x_k\}}{e}}
    \infer1{\valscoped{\Gamma}{\fun{f/k}{x_1}{x_k}{e}}}
    \end{prooftree}
    \hfill
    \begin{prooftree}
        \hypo{\valscoped{\Gamma}{v_1}}
        \hypo{\valscoped{\Gamma}{v_2}}
        \infer2{\valscoped{\Gamma}{[v_1 | v_2]}}
    \end{prooftree}
    
    \vspace{0.3cm}
    \begin{prooftree}
    \infer0{\valscoped{\Gamma}{[]}}
    \end{prooftree}
    \hfill
    \begin{prooftree}
    \hypo{\valscoped{\Gamma}{e}}
    \infer1{\expscoped{\Gamma}{e}}
    \end{prooftree}
    \hfill
    \begin{prooftree}
    \hypo{\expscoped{\Gamma}{e}}
    \hypo{\expscoped{\Gamma}{e_1}}
    \hypo{\cdots}
    \hypo{\expscoped{\Gamma}{e_k}}
    \infer4{\expscoped{\Gamma}{\apply{e}{e_1}{e_k}}}
    \end{prooftree}
    \hfill
    \begin{prooftree}
    \hypo{\expscoped{\Gamma}{e_1}}
    \hypo{\expscoped{\Gamma \cup \{x\}}{e_2}}
    \infer2{\expscoped{\Gamma}{\elet{x}{e_1}{e_2}}}
    \end{prooftree}
    
    \vspace{0.3cm}
    
    \begin{prooftree}
    \hypo{\expscoped{\Gamma \cup \{f/k, x_1, \dots, x_k\}}{e_0}}
    \hypo{\expscoped{\Gamma \cup \{f/k\}}{e}}
    \infer2{\expscoped{\Gamma}{\letrec{f/k}{x_1}{x_k}{e_0}{e}}}
    \end{prooftree}
    \hfill
    \begin{prooftree}
    \hypo{\expscoped{\Gamma}{e_1}}
    \hypo{\expscoped{\Gamma}{e_2}}
    \infer2{\expscoped{\Gamma}{e_1 + e_2}}
    \end{prooftree}
    
    \vspace{0.3cm}
    
    \begin{prooftree}
    \hypo{\expscoped{\Gamma}{e_1}}
    \hypo{\expscoped{\Gamma \cup \textit{vars}(p)}{e_2}}
    \hypo{\expscoped{\Gamma}{e_3}}
    \infer3{\expscoped{\Gamma}{\case{e_1}{p}{e_2}{e_3}}}
    \end{prooftree}
    \hfill
    \begin{prooftree}
        \hypo{\expscoped{\Gamma}{e_1}}
        \hypo{\expscoped{\Gamma}{e_2}}
        \infer2{\expscoped{\Gamma}{[e_1 | e_2]}}
    \end{prooftree}
    
    \caption{Scoping rules}
    \label{fig:scoping}
\end{figure}

The scoping judgement can be generalised to substitutions. This relation characterises that the substitution $\sub$ maps the names in $\Gamma$ to \emph{values} that are scoped in $\Delta$.

\begin{definition}[Substitution scoping]
$\subscoped{\Gamma}{\sub}{\Delta} := \forall x \in \Gamma: \valscoped{\Delta}{\sub(x)}$.
\end{definition}

Based on scoping, a number of lemmas can be proven about substitution; we highlight the most important ones. In some lemmas, we state the same (or similar) properties for both expressions and values, either due to mutual induction or because the property needs to be shown for both types separately. For instance, when proving a general statement on function \emph{values}, the induction hypothesis requires the statement to hold on the body \emph{expression} of the function.

\begin{lemma}[Scoping of extended substitutions]\label{thm:extsubst}
\begin{flalign*}
	&\forall v, \sub, \Delta, \Gamma: \valscoped{\Delta}{v} \land \subscoped{\Gamma}{\sub}{\Delta} \implies \forall x \notin \Gamma: \subscoped{\Gamma \cup \{x\}}{\subst{\sub}{x \mapsto v}}{\Delta}
\end{flalign*}
\end{lemma}
\begin{proof}
	By the definition of substitution update, $\subst{\sub}{x \mapsto v}(x) = v$, which satisfies $\valscoped{\Delta}{v}$ according to the first assumption. On the other hand, for any $y \neq x$, $\subst{\sub}{x \mapsto v} y = \sub(y)$, which also satisfies $\valscoped{\Delta}{\sub(y)}$ according to the second assumption.
\end{proof}

\begin{lemma}[Scope extension (weakening)]\label{thm:scopeExtend}
\begin{flalign*}
  &\forall v, \sub, \Gamma, \Delta: \valscoped{\Gamma}{v} \implies \valscoped{\Gamma \cup \Delta}{v} \\
  &\forall e, \sub, \Gamma, \Delta: \expscoped{\Gamma}{e} \implies \expscoped{\Gamma \cup \Delta}{e}
\end{flalign*}
\end{lemma}
\begin{proof}
	We prove this property by mutual induction on the $\valscoped{\Gamma}{v}$ and $\valscoped{\Gamma}{e}$ judgements.
	\begin{itemize}
		\item For the recursive cases (see \Cref{fig:scoping}), it suffices to apply the rules on the induction hypotheses. For example, if $e = e_1 + e_2$, we get the following induction hypotheses:
		\begin{itemize}
			\item $\valscoped{\Gamma \cup \Delta}{e_1} \land \expscoped{\Gamma \cup \Delta}{e_1}$
			\item $\valscoped{\Gamma \cup \Delta}{e_2} \land \expscoped{\Gamma \cup \Delta}{e_2}$
		\end{itemize}
		We can use these to show that the premises of $\expscoped{\Gamma \cup \Delta}{e_1 + e_2}$ hold.
		\item If $v = l$, then $\valscoped{\Gamma \cup \Delta}{l}$ holds by definition.
		\item If $v = x$ (or $v = f/k$ respectively), then according to the premise $x \in \Gamma$, then by set reasoning $x \in \Gamma \cup \Delta$.
	\end{itemize}
\end{proof}

\begin{lemma}[Scoping of restricted substitutions]\label{thm:restrictScope}
\begin{flalign*}
\forall \sub, \Gamma, \Delta:&\ \subscoped{\Gamma}{\sub}{\Delta} \implies\\
&\subscoped{\Gamma \cup \{x_1, \dots, x_k\}}{\sub \restrict \{x_1, \dots, x_k\}}{\Delta \cup \{x_1, \dots, x_k\}}
\end{flalign*}
\end{lemma}
\begin{proof}
	We do case distinction on whether the name was removed from the substitution.
	\begin{itemize}
		\item For any $x_i$, $(\sub \restrict \{x_1, \dots, x_k\})(x_i) = x_i$ by definition, and clearly $x_i \in \Delta \cup \{x_1, \dots, x_k\}$ thus it satisfies $\valscoped{\Delta \cup \{x_1, \dots, x_k\}}{x_i}$.
		\item For any other $x_i \neq x \in \Gamma$, $(\sub \restrict \{x_1, \dots, x_k\})(x) = \sub(x)$, for which we already know that $\valscoped{\Delta}{\sub(x)}$ by the premise. Now we use \Cref{thm:scopeExtend} with $\{x_1, \dots, x_k\}$ to conclude the proof.
	\end{itemize}
\end{proof}

Now we define the set of variables that are not modified by a substitution.

\begin{definition}[Substitution identities]\ \\\indent
	$\preserves{\Gamma}{\sub} := \forall x \in \Gamma: \sub(x) = \inr{x}$
\end{definition}

Any names removed from a substitution become identity, therefore preserved.

\begin{lemma}[Restriction of substitution identities]\label{thm:restrictIdentities}
	\begin{flalign*}
		&\preserves{\Gamma}{\sub} \implies \preserves{\Gamma \cup \{x_1, \dots, x_k\}}{\sub \restrict \{x_1, \dots, x_k\}}
	\end{flalign*}
\end{lemma}
\begin{proof}
	We need to prove that $\forall x \in \Gamma \cup \{x_1, \dots, x_k\}: (\sub \restrict \{x_1, \dots, x_k\})(x_i) = \inr{x_i}$
	\begin{itemize}
		\item For any $x_i$, $(\sub \restrict \{x_1, \dots, x_k\})(x_i) = \inr{x_i}$ by definition. 
		\item For any other $x \in \Gamma$, $(\sub \restrict \{x_1, \dots, x_k\})(x) = \inr{x}$ since $\preserves{\Gamma}{\sub}$.
	\end{itemize}
\end{proof}

\begin{lemma}[Preserving substitution is identity]\label{thm:preserves}
	\begin{flalign*}
		&\forall v, \sub, \Gamma: \valscoped{\Gamma}{v} \land \preserves{\Gamma}{\sub} \implies \subst{v}{\sub} = v \\
		&\forall e, \sub, \Gamma: \expscoped{\Gamma}{e} \land \preserves{\Gamma}{\sub} \implies \subst{e}{\sub} = e
	\end{flalign*}
\end{lemma}
\begin{proof}
	We prove this property by mutual induction on the $\valscoped{\Gamma}{v}$ and $\valscoped{\Gamma}{e}$ judgements.
	\begin{itemize}
		\item If $v = l$, the substitution does not affect $l$.
		\item If $v = x$ (or $v = f/k$ respectively), there are two cases:
		\begin{itemize}
			\item $x \in \Gamma$: then $\subst{x}{\sub} = \sub(x) = x$, because $\preserves{\Gamma}{\sub}$.
			\item $x \notin \Gamma$: then $\valscoped{\Gamma}{x}$ premise is not satisfied, thus we get a contradiction.
		\end{itemize}
		\item Simple recursive cases ($\apply{e}{e_1}{e_k}$, $\case{e_1}{p}{e_2}{e_3}$,  $e_1 + e_2$) follow from the equalities of the induction hypotheses.
		\item In case of more complex recursive cases ($\mathtt{let}, \mathtt{rec}, \mathtt{fun}$) we use \Cref{thm:restrictIdentities}. We show the proof for $\letrec{f/k}{x_1}{x_k}{e_0}{e}$, the other cases can be constructed similarly.
		The goal is the following:
		\begin{align*}
			\letrec{f/k&}{x_1}{x_k}{\subst{e_0}{\sub \restrict \{f/k, x_1, \dots, x_k\}}}{\subst{e}{\sub \restrict \{f/k\}}} = \\
			&\letrec{f/k}{x_1}{x_k}{e_0}{e}
		\end{align*}
		We have two induction hypotheses:
		\begin{itemize}
			\item $\preserves{\Gamma \cup \{f/k, x_1, \dots, x_k\}}{e_0} \implies \subst{e_0}{\sub \restrict \{f/k, x_1, \dots, x_k\}} = e_0$
			\item $\preserves{\Gamma \cup \{f/k\}}{e} \implies \subst{e}{\sub \restrict \{f/k\}} = e$
		\end{itemize}
		
		By using \Cref{thm:restrictIdentities}, we can prove both premises, then we can conclude the proof with applying the assumptions.
	\end{itemize}
\end{proof}

\begin{corollary}[Closed expressions are not modified by substitutions]\label{thm:closedSubst}
\begin{flalign*}
  &\forall v, \sub: \valscoped{\emptyset}{v} \implies \subst{v}{\sub} = v \\
  &\forall e, \sub: \expscoped{\emptyset}{e} \implies \subst{e}{\sub} = e
\end{flalign*}
\end{corollary}
\begin{proof}
	First, we prove that $\forall \sub, \preserves{\emptyset}{\sub}$, which holds because there are no elements in $\emptyset$. Thereafter, we conclude the proof with using \Cref{thm:preserves} with $\Gamma = \emptyset$, and this way both its premises are satisfied.
\end{proof}

Scoping of values (and expressions) can be combined with the scoping of substitutions: 
applying a scoped substitution on a scoped value (or expression) keeps it scoped, and conversely, a substituted value (or expression) can be shown to be scoped without the substitution.

\begin{lemma}[Substitution preserves scoping]\label{thm:substPres}
\begin{flalign*}
  &\forall e, \Gamma: \expscoped{\Gamma}{e} \implies \forall \Delta, \sub: \subscoped{\Gamma}{\sub}{\Delta} \implies \expscoped{\Delta}{\subst{e}{\sub}} \\
  &\forall v, \Gamma: \valscoped{\Gamma}{v} \implies \forall \Delta, \sub: \subscoped{\Gamma}{\sub}{\Delta} \implies \valscoped{\Delta}{\subst{v}{\sub}}
\end{flalign*}
\end{lemma}
\begin{proof}
	We carry out induction on the syntax of expressions.
	\begin{itemize}
		\item If $v = l$, then it is scoped in any $\Delta$ by definition.
		\item If $v = x$ (or $v = f/k$ resp.), then $x \in \Gamma$ because of $\valscoped{\Gamma}{x}$, and in addition, $\subst{x}{\sub} = \sub(x)$. Thus we prove $\valscoped{\Delta}{\sub(x)}$ by using the other premise ($\subscoped{\Gamma}{\sub}{\Delta}$).
		\item For the simple recursive language elements ($e_1 + e_2, \mathtt{case}, \mathtt{apply}$), we can just use the induction hypotheses to scope the subexpressions, and then use the scoping rule for the specific construct to conclude the proof.
		\item For the more complex language elements ($\mathtt{let}, \mathtt{rec}, \mathtt{fun}$) we also need to involve \Cref{thm:restrictScope} in the proof. For instance, in the case of $\mathtt{fun}$, after using the scoping rule for it, we need to prove that $\expscoped{\Delta \cup \{x_1, \dots, x_k\}}{\subst{e}{\sub \restrict \{x_1, \dots, x_k\}}}$ which can be done by applying \Cref{thm:restrictScope} on the premise $\subscoped{\Gamma}{\sub}{\Delta}$, and then applying the induction hypothesis.
	\end{itemize}
\end{proof}

\begin{lemma}[Substitution implies scoping]\label{thm:substImpl}
\begin{flalign*}
    &\forall e, \Gamma, \Delta: (\forall \sub: \subscoped{\Gamma}{\sub}{\Delta} \implies \expscoped{\Delta}{\subst{e}{\sub}}) \implies \expscoped{\Gamma}{e}\\
    &\forall v, \Gamma, \Delta: (\forall \sub: \subscoped{\Gamma}{\sub}{\Delta} \implies \valscoped{\Delta}{\subst{v}{\sub}}) \implies \expscoped{\Gamma}{v}
\end{flalign*}
\end{lemma}
\begin{proof}
	For this proof, we refer to the formalisation~\cite{coreerlangmini}, because the nameless variable representation plays a key role in it.
\end{proof}

\subsection{Frame Stack Semantics}
\label{sec:framesemantics}

After defining syntax, substitution, and name scoping, we define a frame stack semantics for our language by reusing techniques from Pitts~\cite{pitts2000operational}. This sort of definition resembles reduction style~\cite{felleisen1986control} semantics, but in this case ``the evaluation contexts are decomposed into a list of evaluation frames''~\cite{pitts2000operational}. In other words, the frame stack is essentially the continuation of the computation.


First, we describe the syntax of frame stacks and frames, which resemble syntactical contexts. For the stacks, we use lists and use the following notations: $\idfs$ denotes the empty stack and $[F_1, \dots, F_n]$ denotes the stack containing $F_1, \dots, F_n$ frames in this order, and $\consfs{F}{K}$ prepends the frame $F$ to $K$.

\begin{definition}[Syntax of frames, frame stacks]
\begin{flalign*}
F ::=\ & \apply{\Box}{e_1}{e_k} \hda \apply{v_0}{\Box}{e_k} \hda \cdots \hda \apply{v_0}{v_1}{\Box} \\
& \hda \elet{x}{\Box}{e_2} \hda
  \Box + e_2 \hda v_1 + \Box \hda \case{\Box}{e_2}{e_3}\\
& \hda [e_1 | \Box] \hda [\Box | v_2] \\
K ::=\ &[F_1, \dots, F_n]
\end{flalign*}
\end{definition}

We reuse the notation for substitutions $\subst{F}{e}$ for the substitution of the $\Box$ in frame $F$.

\begin{definition}[Closed frame stacks]
A frame stack $K$ is said to be closed ($\framesclosed{K}$) if and only if all of its frames are closed. A frame is closed, if its constituent values and expressions are closed; formally:
\begin{align*}
\frameclosed{(\apply{\Box}{e_1}{e_k})} &:= \expscoped{\emptyset}{e_1} \land \dots \land \expscoped{\emptyset}{e_k} \\
\frameclosed{(\apply{v_0}{\Box}{e_k})} &:= \valscoped{\emptyset}{v_0} \land \expscoped{\emptyset}{e_2} \dots \land \expscoped{\emptyset}{e_k} \\
\frameclosed{(\apply{v_0}{v_1}{\Box})} &:= \valscoped{\emptyset}{v_0} \land \valscoped{\emptyset}{v_1} \land \dots \land \valscoped{\emptyset}{v_{k-1}}\\
\frameclosed{(\elet{x}{\Box}{e_2})} &:= \expscoped{\{x\}}{e_2}\\
\frameclosed{(\Box + e_2)} &:= \expscoped{\emptyset}{e_2}\\
\frameclosed{(v_1 + \Box)} &:= \valscoped{\emptyset}{v_1}\\
\frameclosed{(\case{\Box}{p}{e_2}{e_3})} &:= \expscoped{\textit{vars}(p)}{e_2} \land \expscoped{\emptyset}{e_3}\\
\frameclosed{[e_1 | \Box]} &:= \expscoped{\emptyset}{e_1}\\
\frameclosed{[\Box | v_2]} &:= \valscoped{\emptyset}{v_2}
\end{align*}
\end{definition}

Now we can define the one-step evaluation relation of the frame stack semantics, followed by the many-step relation given as the reflexive, transitive closure of the one-step relation with step-indexing (see \Cref{fig:step}). We invite the reader to pay attention to the evaluation order determined by the rules, which follows the principles of (Core) Erlang~\cite{neuhausser2007abstraction}.
For the sake of simplifying some equivalence-related notions, we also define any-step evaluation in the following way: $\rewritesstar{K}{e}{v} := \exists n: \rewritesn{K}{e}{n}{\idfs}{v}$.

\begin{figure}
    \begin{flalign*}
    &\rewrites{K}{\elet{x}{e_1}{e_2}}{\elet{x}{\Box}{e_2} :: K}{e_1} \\
    &\rewrites{K}{[e_1|e_2]}{[e_1|\Box] :: K}{e_2} \\
    &\rewrites{K}{\apply{e}{e_1}{e_k}}{\apply{\Box}{e_1}{e_k} :: K}{e} \\
    &\rewrites{K}{e_1 + e_2}{\Box + e_2 :: K}{e_1} \\
    &\rewritesbr{K}{\letrec{f/k}{x_1}{x_n}{e_0}{e}}{K}{\subst{e}{f/k \mapsto \fun{ f/k}{x_1}{x_n}{e_0}}} \\
    &\rewritesbr{K}{\case{e_1}{p}{e_2}{e_3}}{\case{\Box}{p}{e_2}{e_3} :: K}{e_1} \\
    &\rewrites{\apply{\Box}{e_1}{e_k} :: K}{v}{\apply{v}{\Box}{e_k} :: K}{e_1} \\
    &\rewrites{\applyz{\Box} :: K}{\funz{f/0}{e}}{K}{\subst{e}{f/0 \mapsto \funz{f/0}{e}}} \\
    &\rewritesbr{\mathtt{apply\ } v(v_1, \dots, v_{i-1}, \Box, e_{i + 1}, \dots, e_k) :: K}{v_i}{\mathtt{apply\ } v(v_1, \dots, v_{i-1}, v_i, \Box, e_{i + 2}, \dots, e_k) :: K}{e_{i+1}} \qquad(\text{if } i < k)\\
    &\rewritesbr{\apply{(\fun{f/k}{x_1}{x_k}{e})}{v_1}{\Box} :: K}{v_k}{K}{\subst{e}{f/k \mapsto \fun{f/k}{x_1}{x_k}{e}, x_1 \mapsto v_1, \dots, x_k \mapsto v_k}} \\
    &\rewrites{\elet{x}{\Box}{e_2}::K}{v}{K}{e_2[x \mapsto v]} \\
    &\rewrites{[e_1|\Box] :: K}{v_2}{[\Box | v_2] :: K}{e_1}\\
    &\rewrites{[\Box|v_2] :: K}{v_1}{K}{[v_1|v_2]}\\
    &\rewrites{\case{\Box}{p}{e_2}{e_3}::K}{v}{K}{\subst{e_2}{\textit{match}(p,v)}} \\&\hspace{26em}(\text{if } \textit{is\_match}(p,v)) \\
    &\rewrites{\case{\Box}{p}{e_2}{e_3}::K}{v}{K}{e_3} \qquad(\text{if } \neg\textit{is\_match}(p,v)) \\
    &\rewrites{\Box + e_2::K}{v}{v + \Box::K}{e_2} \\
    &\rewrites{l_1 + \Box::K}{l_2}{K}{l_1 + l_2}\\
    \end{flalign*}
    
    \begin{prooftree}
    \infer0{\rewritesn{K}{e}{0}{K}{e}}
    \end{prooftree}
    \hfill
    \begin{prooftree}
    \hypo{\rewrites{K}{e}{K'}{e'}}
    \hypo{\rewritesn{K'}{e'}{n}{K''}{e''}}
    \infer2{\rewritesn{K}{e}{1 + n}{K''}{e''}}
    \end{prooftree}

\caption{Frame stack evaluation relations}
\label{fig:step}
\end{figure}

\section{Naive Definitions of Program Equivalence}\label{sec:naive}

This section we briefly overview naive but imprecise definitions of expression equivalence. The following relations are simply induced by the evaluation relation, and are insufficient  in the case of this higher-order functional language.

\subsection{Behavioural Equivalence}
\label{subsec:behavioral}

Naive program equivalence (simple behavioural equivalence by the textbook~\cite{pierce2010software} definition) says two expressions to be equivalent if and only if in every starting configuration they evaluate to the same \emph{result}, or they both diverge. Note that the equivalence relation is typically defined by the symmetrisation of a preorder relation.

\begin{definition}[Naive behavioural equivalence]
\begin{flalign*}
&e_1 \behpre e_2 := \forall v: \rewritesstar{\idfs}{e_1}{v} \implies \rewritesstar{\idfs}{e_2}{v} \\
&e_1 \behequiv e_2 := e_1 \behpre e_2 \land e_2 \behpre e_1
\end{flalign*}
\end{definition}

In the formalisation, we prove the $\behequiv$ relation to be an equivalence, that is, show that it is reflexive, symmetric, and transitive. We also show that it is indeed a behavioural equivalence over our expressions, characterised by congruence~\cite{pierce2010software} property. Congruence helps prove compound expressions equivalent, but it took significant effort to formalise and prove in Coq~\cite{coreerlang}.

\paragraph{Equivalence of function expressions}

The definition of behavioural equivalence checks for equality of result values. In case of integer or list expressions, it is likely to relate expressions understood equivalent. However, for function expressions, checking strict equality is likely not to meet our expectations: 
\begin{flalign*}
	\mathtt{fun}\ f/1(\text{X}) \rightarrow \text{X} + 2 = \mathtt{fun}\ f/1(\text{X}) \rightarrow (\text{X} + 1) + 1 
\end{flalign*}

Obviously, we would expect functions with equivalent bodies to be equivalent, but as exemplified by the above snippets, the naive approach only relates function closures whose body expressions are structurally equal. With this naive definition, we can only prove identical function expressions equivalent.

\subsection{Naive Contextual Equivalence}
\label{subsec:ctx}

In case of function expressions, the equivalence relation should depend on whether the body expressions behave the same way in the same \emph{expression contexts}, i.e. they are contextually equivalent. In other words, contexts are supposed to reveal any observable differences between the expressions.

We proceed by defining expression contexts, which are basically expressions with one of their subexpressions replaced by a \emph{hole}\footnote{We only consider linear (also called single-hole) contexts.}. The hole can be substituted by any expression (e.g. $\subst{\mathtt{apply}\ f(1, \Box, 3)}{2} = \mathtt{apply}\ f(1, 2, 3)$) to obtain valid expressions. We define expression contexts as follows:

\begin{definition}[Expression context]
\begin{flalign*}
C ::=&\ \Box \hda \fun{f/k}{x_1}{x_k}{C} \hda \elet{x}{C}{e_2} \hda \elet{x}{e_1}{C}  \\
&\hda \apply{C}{e_1}{e_k} \hda \apply{e}{C}{e_k} \hda \dots \hda \apply{e}{e_1}{C} \\
& \hda \letrec{f/k}{x_1}{x_k}{C}{e} \hda \letrec{f/k}{x_1}{x_k}{e_0}{C}\\
&\hda C + e_2 \hda e_1 + C \hda [C|e_2] \hda [e_1|C] \hda \case{C}{p}{e_2}{e_3} \\
&\hda \case{e_1}{p}{C}{e_3} \hda \case{e_1}{p}{e_2}{C}
\end{flalign*}
\end{definition}

We define the substitution of the $\Box$ by an expression in the usual way (e.g. see~\cite{pitts2000operational}), and we will denote it by $\subst{C}{e}$. Next, we define the contextual equivalence using contextual preorders.

\begin{definition}[Naive contextual equivalence]
\begin{align*}
    &e_1 \ctxpre e_2 := \forall C, v: 
    \rewritesstar{\idfs}{\subst{C}{e_1}}{v} \implies \rewritesstar{\idfs}{\subst{C}{e_2}}{v}\\
    &e_1 \ctxequiv e_2 := e_1 \ctxpre e_2 \land e_2 \ctxpre e_1
\end{align*}
\end{definition}

Let us point out that this definition of naive contextual equivalence cannot overcome the issue with function expressions: even if we consider syntactical contexts instead of arbitrary frame stacks, the value (i.e. the function closure) is checked for equality. Actually, it escalates the problem even further: the latter relation coincides with syntactical equality since it requires the expressions to yield equal values, even in function abstraction contexts. Clearly, a special notion of value equality is needed to treat function expressions properly.


In addition, while it is straightforward to disprove expressions contextually equivalent, the proof of equivalence in general is understood to be significantly more complex as it requires induction over contexts. In order to overcome these issues, we seek for equal but differently formulated equivalence relations, which coincide with our intuition and at the same time they ease proving two expressions equivalent.

\section{Precise Program Equivalence Definitions}
\label{sec:equiv}

In order to be able to reason about the correctness of refactorings, an appropriate and precise program equivalence definition is needed for the object language. In this section, we first refine the contextual equivalence relation introduced in the previous section, then we present a number of alternative, equal definitions, such as step-indexed logical relations~\cite{wand2018contextual,pitts1998operational,pitts2000operational} and CIU (``closed instances of uses'') relations.

\subsection{Frame Stack Termination Relation}
\label{sec:term}

As briefly mentioned in \Cref{sec:related}, it is an important result discussed in related work that two expressions can be shown contextually equivalent by proving that they both terminate or both diverge in arbitrary contexts. In other words, it suffices to prove this termination property even when arguing about expression equivalences appearing in the verification of refactoring. Thus, to be able to reason about termination easily, we first formalise an inductive, step-indexed termination relation for our case study language (see \Cref{fig:termrel}). We also introduce any-step termination: $\term{K}{e} := \exists n: \termk{K}{e}{n}$.

\begin{figure}[htb]
    \centering
    \begin{prooftree}
    \infer0{\termk{\idfs}{v}{0}}
    \end{prooftree}
    \hfill
    \begin{prooftree}
    \hypo{\termk{\case{\Box}{p}{e_2}{e_3} :: K}{e}{n}}
    \infer1{\termk{K}{\case{e}{p}{e_2}{e_3}}{1 + n}}
    \end{prooftree}
    \hfill
    \begin{prooftree}
    \hypo{\termk{\Box + e_2 :: K}{e}{n}}
    \infer1{\termk{K}{e + e_2}{1 + n}}
    \end{prooftree}
    
    \vspace{0.25cm}
    
    \begin{prooftree}
    \hypo{\termk{\elet{x}{\Box}{e_2} :: K}{e}{n}}
    \infer1{\termk{K}{\elet{x}{e}{e_2}}{1 + n}}
    \end{prooftree}
    \hfill
    \begin{prooftree}
    \hypo{\termk{K}{\subst{e_2}{\textit{match}(p, v)}}{n}}
    \hypo{\textit{is\_match}(p, v)}
    \infer2{\termk{\case{\Box}{p}{e_2}{e_3} :: K}{v}{1 + n}}
    \end{prooftree}
    
    \vspace{0.25cm}
    
    \begin{prooftree}
    \hypo{\termk{K}{e_3}{n}}
    \hypo{\neg\textit{is\_match}(p, v)}
    \infer2{\termk{\case{\Box}{p}{e_2}{e_3} :: K}{v}{1 + n}}
    \end{prooftree}
    
    \vspace{0.25cm}
    \begin{prooftree}
    \hypo{\termk{[e_1 | \Box ] :: K}{e_2}{n}}
    \infer1{\termk{K}{[e_1 | e_2]}{1 + n}}
    \end{prooftree}
    \hfill
    \begin{prooftree}
    \hypo{\termk{[\Box | v_2 ] :: K}{e_1}{n}}
    \infer1{\termk{[e_1 | \Box ] :: K}{v_2}{1 + n}}
    \end{prooftree}
    \hfill
    \begin{prooftree}
    \hypo{\termk{K}{[v_1 | v_2]}{n}}
    \infer1{\termk{[\Box | v_2 ] :: K}{v_1}{1 + n}}
    \end{prooftree}
    \vspace{0.25cm}
    
    \begin{prooftree}
    \hypo{\termk{v + \Box :: K}{e_2}{n}}
    \infer1{\termk{\Box + e_2 :: K}{v}{1 + n}}
    \end{prooftree}
    \hfill
    \begin{prooftree}
    \hypo{\termk{K}{l_1 + l_2}{n}}
    \infer1{\termk{l_1 + \Box :: K}{l_2}{1 + n}}
    \end{prooftree}
    \hfill
    \begin{prooftree}
    \hypo{\termk{K}{e_2[x \mapsto v]}{n}}
    \infer1{\termk{\elet{x}{\Box}{e_2}:: K}{v}{1 + n}}
    \end{prooftree}
    
    \vspace{0.25cm}
    
    \begin{prooftree}
    \hypo{\termk{K}{\subst{e}{f/k \mapsto \fun{f/k}{x_1}{x_k}{e_0}}}{n}}
    \infer1{\termk{K}{\letrec{f/k}{x_1}{x_k}{e_0}{e}}{1 + n}}
    \end{prooftree}
    
    \vspace{0.25cm}
    
    \begin{prooftree}
   	\hypo{\termk{\apply{\Box}{e_1}{e_k} :: K}{e}{n}}
   	\infer1{\termk{K}{\apply{e}{e_1}{e_k})}{1 + n}}
    \end{prooftree}
	\hfill
    \begin{prooftree}
    \hypo{\termk{\apply{v}{\Box}{e_k}) :: K}{e_1}{n}}
    \infer1{\termk{\apply{\Box}{e_1}{e_k}) :: K}{v}{1 + n}}
    \end{prooftree}
    
    \vspace{0.25cm}
    For empty parameter list, an additional rule need to be introduced:
    \vspace{0.25cm}
    
    \begin{prooftree}
    \hypo{\termk{K}{\subst{b}{f/0 \mapsto \funz{f/0}{e_0}}}{n}}
    \infer1{\termk{\applyz{\Box} :: K}{\funz{f/0}{e_0}}{1 + n}}
    \end{prooftree}
    
    \vspace{0.25cm}
    
    \begin{prooftree}
    \hypo{\termk{\mathtt{apply}\ v(v_1, \dots, v_i, \Box, e_{i+2}, \dots, e_k) :: K}{e_{i+1}}{n}}
    \infer1{\termk{\mathtt{apply}\ v(v_1, \dots, v_{i-1}, \Box, e_{i+1}, \dots e_k) :: K}{v_i}{1 + n}}
    \end{prooftree}
    
    \vspace{0.25cm}
    
    \begin{prooftree}
    \hypo{\termk{K}{\subst{b}{f/k \mapsto \fun{f/k}{x_1}{x_k}{e_0}, x_1 \mapsto v_1, \dots, x_k \mapsto v_k}}{n}}
    \infer1{\termk{\apply{(\fun{f/k}{x_1}{x_k}{e_0})}{v_1}{\Box} :: K}{v_k}{1 + n}}
    \end{prooftree}
    
    \caption{Step-indexed, frame stack-style termination relation}
    \label{fig:termrel}
\end{figure}

We can prove the following correspondences between step-indexed evaluation and the termination relation.

\begin{lemma}[Step-indexed terminations coincide]\label{thm:semtermEqTerm}
\begin{flalign*}
	&\termk{K}{e}{n} \iff \exists v: \rewritesn{K}{e}{n}{\idfs}{v}
\end{flalign*}
\end{lemma}
\begin{proof}
This lemma can be proven by induction on the step-index. The subgoal can be solved by basically applying the induction hypotheses for the subexpressions. For more details, we refer to the Coq implementation~\cite{coreerlangmini}.
\end{proof}

\begin{corollary}[Terminations coincide]
	\begin{flalign*}
		&\term{K}{e} \iff \exists v: \rewritesstar{K}{e}{v}
	\end{flalign*}
\end{corollary}

We have investigated several properties of the semantics relevant to the termination relation; here we highlight some of these, the rest can be found in the Coq formalisation. In the following two lemmas, we reuse the notation $\subst{F}{e}$ to substitute the $\Box$ in frame $F$ by $e$, the operation $::$ prepends a frame to a stack, and $\concat$ is concatenation between frame stacks.

\begin{lemma}[Remove frame]\label{thm:putback} For any closed frame $F$,
\begin{flalign*}
	&\term{F :: K}{e} \implies \term{K}{\subst{F}{e}}
\end{flalign*}
\end{lemma}
\begin{proof}
From $\term{F :: K}{e}$ we assume that this termination takes $k$ steps. We proceed with case distinction on frame $F$: depending on the structure of $F$, the evaluation of $\subst{F}{e}$ should take $n + k$ steps (e.g. for $F = \Box + e_2$, $n = 1$, while for $F = v_1 + \Box$, $n = 2$, etc.) to reach the configuration in the premise ($\termk{F :: K}{e}{k}$).
\end{proof}

\begin{lemma}[Add frame]\label{thm:putbackrev} For any closed frame $F$, closed expression $e$,
\begin{flalign*}
	&\term{K}{\subst{F}{e}} \implies \term{F :: K}{e}
\end{flalign*}
\end{lemma}
\begin{proof}
	This proof is the reverse of the previous one. Again, we do case distinction on $F$, and then we inspect the premise $\term{K}{\subst{F}{e}}$ and investigate how this derivation could have been done. After taking some steps (e.g. for $F = \Box + e_2$, one step is enough, while for $F = v_1 + \Box$, two steps are needed, etc.) we reach the configuration with some $k$ number $\termk{F :: K}{e}{k}$ we need to prove to terminate.
\end{proof}

The following lemma is not closely related with termination, but will be important when establishing connection between observational equivalence and other expression equivalence relations.

\begin{lemma}[Extend frame stack]\label{thm:extendframes}
\begin{flalign*}
	&\rewritesn{K_1}{e_1}{n}{K_2}{e_2} \implies \forall K': \rewritesn{K_1 \concat K'}{e_1}{n}{K_2 \concat K'}{e_2}
\end{flalign*}
\end{lemma}
\begin{proof}
	We carry out induction on the length of the derivation ($n$).
	\begin{itemize}
		\item For $n = 0$, from $\rewritesn{K_1}{e_1}{0}{K_2}{e_2}$ we acquire $K_1 = K_2$. Thereafter, we use definition for $\longrightarrow^0$.
		\item For $n = 1 + n'$, we inspect the possible derivations of $\rewritesn{K_1}{e_1}{1 + n'}{K_2}{e_2}$, and just take the same step in the conclusion together with the induction hypothesis (if necessary).
	\end{itemize}
\end{proof}

\subsection{The Logical Relation}
\label{sec:logrel}

The majority of related work on program equivalence proposes logical relations, amongst others, for arguing about standard contextual equivalence. At first we followed the footsteps of Pitts~\cite{pitts2000operational} and adapted his ``logical simulation relation''. Unfortunately, their mathematical definitions cannot be directly formalised in Coq (for an untyped language) as they use statements that are not strictly positive and therefore do not pass Coq's positivity checker. Neither could we use \emph{types} like Culpepper and Cobb did~\cite{culpepper2017contextual} as our language is untyped. Finally, we decided to adopt the idea of step-indexed relations~\cite{wand2018contextual,pitts2010step}.

First, we define these logical relations for closed values, expressions and frame-stacks. For better readability, we omit the assumptions of closedness from the definitions. We invite the reader to observe how the value relation addresses the previously seen issue of function expression equivalence by relating the body expressions.

\begin{definition}[Logical relations for closed expressions, values and frame stacks]\label{def:logrelclosed}
\begin{gather*}
    (l_1, l_2) \in \mathbb{V}_n \iff l_1 = l_2\\
    ([], []) \in \mathbb{V}_n\iff \textit{true}\\
    (\fun{f/k}{x_1}{x_k}{e}, \fun{f/k}{x_1}{x_k}{e'}) \in \mathbb{V}_n \iff \\
    \forall m < n: \forall v_1, v_1', \dots, v_k, v_k': (v_1, v_1') \in \mathbb{V}_m \land \dots  \land (v_k, v_k') \in \mathbb{V}_m \implies\\
    (\subst{e}{f/k \mapsto \fun{f/k}{x_1}{x_k}{e}, x_1 \mapsto v_1, \dots, x_k \mapsto v_k},\\
    \subst{e'}{f/k \mapsto \fun{f/k}{x_1}{x_k}{e'}, x_1 \mapsto v_1', \dots, x_k \mapsto v_k'}) \in \mathbb{E}_m\\
    ([v_1 | v_2], [v_1'| v_2']) \in \mathbb{V}_n \iff (v_1, v_1') \in  \mathbb{V}_n \land (v_2, v_2') \in  \mathbb{V}_n
\end{gather*}
\begin{gather*}
    (K_1, K_2) \in \mathbb{K}_n \iff
    \forall m \le n, v_1, v_2: (v_1, v_2) \in \mathbb{V}_m \implies
    \\ \termk{K_1}{v_1}{m} \implies \term{K_2}{v_2} \\
    (e_1, e_2) \in \mathbb{E}_n \iff
    \forall m \le n, K_1, K_2: (K_1, K_2) \in \mathbb{K}_m \implies
    \\ \termk{K_1}{e_1}{m} \implies \term{K_2}{e_2}
\end{gather*}
\end{definition}

To ensure the well-foundedness of these relations, we used the step-index, which is decreased in $\mathbb{V}$ for functions. For list values on the other hand, we did not decrease this index, we only used structural recursion. Alternatively, the step-index can be reduced in this case too, but then the mechanism of pattern matching needs to be formalised in a step-indexed way too.

Just like in the work of Wand et al.~\cite{wand2018contextual}, these relations with higher indices can differentiate more expressions, values and stacks, i.e. $\mathbb{V}_0 \subseteq \mathbb{V}_1 \subseteq \dots \subseteq \mathbb{V}_{n-1} \subseteq \mathbb{V}_n$ (also for $\mathbb{E}_n$ and $\mathbb{K}_n$).
The above relations can be generalised to open expressions (or values) with closing substitutions (i.e. all free variables of the expression replaced by closed values).

\begin{definition}[Logical relations with closing substitutions]\label{def:logrelopen}
\begin{flalign*}
(\sub_1, \sub_2) \in \mathbb{G}^\Gamma_n \iff & \subscoped{\Gamma}{\sub_1}{\emptyset} \land \subscoped{\Gamma}{\sub_2}{\emptyset} \land
\forall x \in \Gamma: (\sub_1(x), \sub_2(x)) \in \mathbb{V}_n\\
(v_1, v_2) \in \mathbb{V}^\Gamma \iff & \valscoped{\Gamma}{v_1} \land \valscoped{\Gamma}{v_2} \land \forall n, \sub_1, \sub_2: \\
&(\sub_1, \sub_2) \in \mathbb{G}^\Gamma_n \implies (v_1[\sub_1], v_2[\sub_2]) \in \mathbb{V}_n\\
(e_1, e_2) \in \mathbb{E}^\Gamma \iff & \expscoped{\Gamma}{e_1} \land \expscoped{\Gamma}{e_2} \land\forall n, \sub_1, \sub_2: \\
&(\sub_1, \sub_2) \in \mathbb{G}^\Gamma_n \implies (e_1[\sub_1], e_2[\sub_2]) \in \mathbb{E}_n
\end{flalign*}
\end{definition}

After having these relations defined, we proceeded to prove the two most important properties~\cite{culpepper2017contextual,wand2018contextual,pitts2010step} of them: the ``fundamental property'' (a form of reflexivity) and the compatibility rules which are forms of congruence. In our formalisation, we state and prove a number of lemmas that support the proof of the main theorems (we refer to the implementation~\cite{coreerlangmini} for more details).

\begin{theorem}[Compatibility rules]\label{thm:compat}\ \\ \normalfont
    \begin{center}
    \begin{prooftree}
    \hypo{(v, v') \in \mathbb{V}^\Gamma}
    \infer1{(v, v') \in \mathbb{E}^\Gamma}
    \end{prooftree}
    \hfill
    \begin{prooftree}
    \hypo{x \in \Gamma}
    \infer1{(x, x) \in \mathbb{V}^\Gamma}
    \end{prooftree}
    \hfill
    \begin{prooftree}
    \hypo{f/k \in \Gamma}
    \infer1{(f/k, f/k) \in \mathbb{V}^\Gamma}
    \end{prooftree}
    \hfill
    \begin{prooftree}
    \infer0{(l, l) \in \mathbb{V}^\Gamma}
    \end{prooftree}
    \hfill
    \begin{prooftree}
        \infer0{([], []) \in \mathbb{V}^\Gamma}
    \end{prooftree}
    
    \vspace{0.3cm}
    
    \begin{prooftree}
    \hypo{(e_1, e_2) \in \mathbb{E}^{\Gamma \cup \{f/k, x_1, \dots, x_k\}}}
    \infer1{(\fun{f/k}{x_1}{x_k}{e}, \fun{f/k}{x_1}{x_k}{e_2}) \in \mathbb{V}^\Gamma}
    \end{prooftree}
    
    \vspace{0.3cm}

    \begin{prooftree}
    \hypo{(e_1, e_1') \in \mathbb{E}^\Gamma}
    \hypo{(e_2, e_2') \in \mathbb{E}^\Gamma}
    \infer2{(e_1 + e_2, e_1' + e_2') \in \mathbb{E}^\Gamma}
    \end{prooftree}
    \hfill
    \begin{prooftree}
    \hypo{(e_1, e_1') \in \mathbb{E}^\Gamma}
    \hypo{(e_2, e_2') \in \mathbb{E}^{\Gamma \cup \{x\}}}
    \infer2{(\elet{x}{e_1}{e_2}, \elet{x}{e_1'}{e_2'}) \in \mathbb{E}^\Gamma}
    \end{prooftree}
    
    \vspace{0.3cm}
    
    \begin{prooftree}
    \hypo{(e, e') \in \mathbb{E}^\Gamma}
    \hypo{(e_1, e_1') \in \mathbb{E}^\Gamma}
    \hypo{\cdots}
    \hypo{(e_k, e_k') \in \mathbb{E}^\Gamma}
    \infer4{(\apply{e}{e_1}{e_k}, \apply{e'}{e_1'}{e_k'}) \in \mathbb{E}^\Gamma}
    \end{prooftree}
    
    \vspace{0.3cm}
    
    \begin{prooftree}
    \hypo{(e_1, e_1') \in \mathbb{E}^\Gamma}
    \hypo{(e_2, e_2') \in \mathbb{E}^\Gamma}
    \infer2{([e_1 | e_2], [e_1' | e_2']) \in \mathbb{E}^\Gamma}
    \end{prooftree}
    
    \vspace{0.3cm}
    
    \begin{prooftree}
    \hypo{(e, e') \in \mathbb{E}^{\Gamma \cup \{f/k\}}}
    \hypo{(b, b') \in \mathbb{E}^{\Gamma \cup \{f/k, x_1, \dots, x_k\}}}
    \infer2{(\letrec{f/k}{x_1}{x_k}{b}{e}, \letrec{f/k}{x_1}{x_k}{b'}{e'}) \in \mathbb{E}^\Gamma}
    \end{prooftree}
    
    \vspace{0.3cm}
    
    \begin{prooftree}
    \hypo{(e_1, e_1') \in \mathbb{E}^\Gamma}
    \hypo{(e_2, e_2') \in \mathbb{E}^{\Gamma \cup \textit{vars}(p)}}
    \hypo{(e_3, e_3') \in \mathbb{E}^\Gamma}
    \infer3{(\case{e_1}{p}{e_2}{e_3}, \case{e_1'}{p}{e_2'}{e_3'}) \in \mathbb{E}^\Gamma}
    \end{prooftree}
    \end{center}
\end{theorem}
\begin{proof}
	The compatibility rules for the non-recursive language constructs follow from the definitions. For recursive functions, induction was needed by the step-index.
	For the other cases, we give a representative proof outline with the compatibility proof of addition:
	\begin{enumerate}
	\item Give proof for closed expressions, then the compatibility with the closing substitutions is just a consequence of it.
	\item From the premise $\termk{K}{e_1 + e_2}{m}$ (for any $m \le n$), we can deduce $\termk{\Box + e_2 :: K}{e_1}{m - 1}$ from the definition of the termination. For the other derivation, we have $(K, K') \in \mathbb{K}_n$. To get $\term{K'}{e_1' + e_2'}$ it is sufficient to prove $\term{\Box + e_2' :: K'}{e_1'}$. Now, we apply the premise $(e_1, e_1') \in \mathbb{E}_n$, to conclude this sub-proof, but we still need to show that the two frame stacks are in relation.
	\item To prove $(\Box + e_2 :: K, \Box + e_2' :: K') \in \mathbb{K}_m$, we need to prove that for any $k \le m$, $(v_1, v_1') \in \mathbb{V}_m: \termk{\Box + e_2 :: K}{v_1}{k} \implies \term{\Box + e_2' :: K'}{v_1'}$. By definition, we can transform both the premise and the conclusion (just like above): $\termk{v_1 + \Box :: K}{e_2}{k - 1} \implies \term{v_1' + \Box :: K'}{e_2'}$, which can be proven by $(e_2, e_2') \in \mathbb{E}_n$, since $k - 1 < k \le m \le n$. However, this step introduces another premise of the new frame stacks to be in relation.
	\item To prove $(v_1 + \Box :: K, v_1' + \Box :: K') \in \mathbb{K}_m$, we have to go through the same steps as before. We need to prove that $j \le k$, $(v_2, v_2') \in \mathbb{V}_k: \termk{v_1 + \Box :: K}{v_2}{j} \implies \term{v_1' + \Box :: K'}{v_2'}$. Once again, we apply the definition of the termination to get $\termk{K}{v_1 + v_2}{j - 1} \implies \term{K'}{v_1' + v_2'}$. At this point, we can notice, that $v_1$ and $v_2$ have to be integer literals ($l_1, l_2$), otherwise the addition had not been terminated. Moreover, $v_1' = l_1$ and $v_2' = l_2$, otherwise they could not have been in relation ($\mathbb{V}_m$, or $\mathbb{V}_k$ resp.) with their counterparts. Hence, we get the statement $\termk{K}{l_1 + l_2}{j - 1} \implies \term{K'}{l_1 + l_2}$, which is proved by the hypothesis of $(K, K') \in \mathbb{K}_n$.
	\end{enumerate}
	For the complete proof we refer to the formalisation~\cite{coreerlangmini}.
\end{proof}

The fundamental property of the logical relations is just a consequence of the compatibility rules.

\begin{theorem}[Fundamental property]\label{thm:fundamental}\normalfont
    \begin{flalign*}
    \expscoped{\Gamma}{e} \implies& (e, e) \in \mathbb{E}^\Gamma \\
    \valscoped{\Gamma}{v} \implies& (v, v) \in \mathbb{V}^\Gamma \\
    \subscoped{\Gamma}{\sub}{\emptyset} \implies& \forall n: (\sub, \sub) \in \mathbb{G}^\Gamma_n
    \end{flalign*}
\end{theorem}
\begin{proof}
	We carry out induction on $e$ (and $v$ resp.). Then for all cases, we can just use the corresponding compatibility rule from \Cref{thm:compat}, moreover, the premises of these rules are satisfied either by the scoping premises or the induction hypotheses.
	
	The fundamental property of $\mathbb{G}^\Gamma_n$ follows from the fundamental property of $\mathbb{V}^\Gamma$.
\end{proof}

\subsection{CIU Equivalence}

Alongside proving the properties of the logical relations, we have also formalised CIU (``closed instances of uses'') preorder and equivalence relations~\cite{mason1991equivalence}. Informally, two expressions are CIU equivalent if they both terminate or diverge when placed in arbitrary reduction contexts.

\begin{definition}[CIU preorder]
\begin{flalign*}
  e_1 \ciupre e_2 &:= \expscoped{\emptyset}{e_1}  \land \expscoped{\emptyset}{e_2} \land \forall K: \framesclosed{K} \implies \term{K}{e_1} \implies \term{K}{e_2}\\
  e_1 \ciuequiv e_2 &:= e_1 \ciupre e_2 \land e_2 \ciupre e_1 \\
  e_1 \ciupre^\Gamma e_2 &:= \expscoped{\Gamma}{e_1} \land \expscoped{\Gamma}{e_2} \land \forall \sub: \subscoped{\Gamma}{\sub}{\emptyset} \implies \subst{e_1}{\sub} \ciupre \subst{e_2}{\sub} \\
  e_1 \ciuequiv^\Gamma e_2 &:= e_1 \ciupre^\Gamma e_2 \land e_2 \ciupre^\Gamma e_1
\end{flalign*}
\end{definition}

Practice shows, that usually proving expression CIU equivalent is simpler than proving them contextually equivalent~\cite{ramanujam1998foundations} or related by the logical relations (it requires only one frame stack and one substitution, rather than related pairs)~\cite{wand2018contextual}. After defining the CIU preorder, we also proved its correspondence with the logical relations (see~\cite{coreerlangmini}):

\begin{theorem}[CIU coincides with the logical relations]\label{thm:ciulogrel}
\begin{flalign*}
	e_1 \ciupre^\Gamma e_2 \iff (e_1, e_2) \in \mathbb{E}^\Gamma
\end{flalign*}
\end{theorem}
\begin{proof}
We follow the footsteps of Wand et al.~\cite{wand2018contextual} in this proof.

$\Rightarrow$: We prove $(e_1, e_2) \in \mathbb{E}^\Gamma \land e_2 \ciupre^\Gamma e_3 \implies (e_1, e_3) \in \mathbb{E}^\Gamma$, which is a trivial consequence of the definitions. Thereafter, we prove our goal by using $(e_1, e_1) \in \mathbb{E}^\Gamma$ as the first premise of this helper statement by the fundamental property (\Cref{thm:fundamental}).

$\Leftarrow$: The closing substitution required by the CIU preorder is denoted by $\sub$. We specialize \Cref{def:logrelopen} of $\mathbb{E}^\Gamma$ with $\sub_1 = \sub_2 = \sub$, and by the fundamental property (\Cref{thm:fundamental}), $(\sub, \sub) \in \mathbb{G}^\Gamma_n$. Thereafter, we just use \Cref{def:logrelopen} of $\mathbb{E}_n$ to finish the proof.
\end{proof}

\subsection{Example simple equivalences}

We also proved a number of simple programs to be equivalent. We show the proof of the first one, but omit the others and refer to the formalisation~\cite{coreerlangmini}. The first equivalence is special, because it will be used in the proofs for the equality of the equivalence relations. 

\begin{example}[Beta reduction 1]\label{ex:beta1}
\begin{flalign*}
	&\expscoped{\Gamma \cup \{x\}}{e} \land \valscoped{\Gamma}{v} \implies e[x \mapsto v] \ciuequiv^\Gamma \elet{x}{v}{e}
\end{flalign*}
\end{example}
\begin{proof}
	Since we prove an equivalence, it means two preorders. We also need to prove a number of closedness properties, which can be done by the hypotheses and the lemmas in \Cref{sec:subst}, we leave these to the reader.
	\begin{itemize}
		\item First, we need to prove the following: for any closed frame stack $K$ and closing substitution $\sub$ (i.e. $\subscoped{\Gamma}{\sub}{\emptyset}$), $\term{K}{\subst{\subst{e}{x \mapsto v}}{\sub}} \implies \term{K}{\elet{x}{\subst{v}{\sub}}{\subst{e}{\sub}}}$. Let us assume, that $\termk{K}{\subst{\subst{e}{x \mapsto v}}{\sub}}{k}$ for a step-index $k$. We can show, that $\termk{K}{\elet{x}{\subst{v}{\sub}}{\subst{e}{\sub}}}{2 + k}$ by definition. If we make these two steps, we get $\termk{K}{\subst{\subst{e}{\sub \restrict x}}{x \mapsto \subst{v}{\sub}}}{k}$, and by the properties of capture-avoiding substitution, $\subst{\subst{e}{\sub \restrict x}}{x \mapsto \subst{v}{\sub}} = \subst{\subst{e}{x \mapsto v}}{\sub}$ (we refer to the formalisation for more details~\cite{coreerlangmini}).
		\item Next, we need to prove the following: for any closed frame stack $K$ and closing substitution $\sub$ (i.e. $\subscoped{\Gamma}{\sub}{\emptyset}$), $\term{K}{\elet{x}{\subst{v}{\sub}}{\subst{e}{\sub}}} \implies \term{K}{\subst{\subst{e}{x \mapsto v}}{\sub}}$. Now we inspect the premise $\term{K}{\elet{x}{\subst{v}{\sub}}{\subst{e}{\sub}}}$, and conclude, that by definition $\termk{K}{\subst{\subst{e}{\sub \restrict x}}{x \mapsto \subst{v}{\sub}}}{k}$ should hold for some $k$. This $k$ is suitable for the derivation in the goal ($\termk{K}{\subst{\subst{e}{x \mapsto v}}{\sub}}{k}$), which is identical to this premise when we use the previous thought about the equality of the substitutions $\subst{\subst{e}{\sub \restrict x}}{x \mapsto \subst{v}{\sub}} = \subst{\subst{e}{x \mapsto v}}{\sub}$.
	\end{itemize}
\end{proof}

\begin{example}[Beta reduction 2]
\begin{flalign*}
	&\valscoped{\Gamma}{\fun{f/k}{x_1}{x_k}{e}} \land \valscoped{\Gamma}{v_1} \land \dots \land \valscoped{\Gamma}{v_k} \implies\\
	&\qquad\subst{e}{f/k \mapsto \fun{f/k}{x_1}{x_k}{e}, x_1 \mapsto v_1, \dots, x_k \mapsto v_k} \ciuequiv^\Gamma\\
	&\qquad\apply{(\fun{f/k}{x_1}{x_k}{e})}{v_1}{v_k}
\end{flalign*}
\end{example}

\begin{example}[Beta reduction 3]
\begin{flalign*}
	&\valscoped{\emptyset}{v_1} \land \dots \valscoped{\emptyset}{v_k} \land \expscoped{\Gamma}{e} \land x_1, \dots, x_k \notin \Gamma \implies\\
	&\qquad e \ciuequiv^\Gamma \apply{(\fun{f/k}{x_1}{x_k}{e})}{v_1}{v_k}
\end{flalign*}
\end{example}

The following two equivalences are also special to us: in our current language (which is a simplified variant of sequential Core Erlang) they hold, but with side effects and exceptions added, other preconditions will be needed to prove them.

\begin{example}[Commutativity of addition]
\begin{flalign*}
&\expscoped{\Gamma}{e_1} \land \expscoped{\Gamma}{e_2} \implies e_1 + e_2 \ciuequiv^\Gamma e_2 + e_1
\end{flalign*}
\end{example}

\begin{example}[Sequencing]
\begin{flalign*}
&\expscoped{\Gamma}{e_2} \land \term{\idfs}{e_1} \land\ \expscoped{\emptyset}{e_1} \land x \notin \Gamma \implies e_2 \ciuequiv^\Gamma \elet{x}{e_1}{e_2}
\end{flalign*}
\end{example}

\subsection{Revisiting Contextual Preorder and Equivalence}

We describe a refined contextual equivalence relation based on the definitions by Wand et al.~\cite{wand2018contextual} and Gordon et al.~\cite{gordon1999compilation}.


\begin{definition}[Contextual preorder]
We define the contextual preorder to be the largest family of relations $R^\Gamma$ that satisfy the following properties:
\begin{itemize}
	\item Adequacy: $(e_1, e_2) \in R^{\emptyset} \implies \term{\idfs}{e_1} \implies \term{\idfs}{e_2}$.
	\item Reflexivity: $(e, e) \in R^\Gamma$.
	\item Transitivity: $(e_1, e_2) \in R^\Gamma \land (e_2, e_3) \in R^\Gamma \implies (e_1, e_3) \in R^\Gamma$.
	\item Compatibility: $R^\Gamma$ satisfies the compatibility rules for every expression from~\Cref{thm:compat}.
\end{itemize}
\end{definition}

We also adjusted our previous notion of contextual preorder and equivalence, and proved that it satisfies the criteria above (we refer to the formalisation~\cite{coreerlangmini}). In this case the context ``closes'' the potentially open expressions.

\begin{flalign*}
e_1 \ctxpre^\Gamma e_2 :=\ & \expscoped{\Gamma}{e_1} \land \expscoped{\Gamma}{e_2} \land
 \forall (C : \textit{Context}): \expscoped{\emptyset}{\subst{C}{e_1}} \land \expscoped{\emptyset}{\subst{C}{e_2}} \implies\\ & \term{\idfs}{\subst{C}{e_1}} \implies \term{\idfs}{\subst{C}{e_2}} \\
e_1 \ctxequiv^\Gamma e_2 :=\ & e_1 \ctxpre^\Gamma e_2 \land e_2 \ctxpre^\Gamma e_1
\end{flalign*}

After defining the contextual preorder, we proved the equality between $\ciupre^\Gamma$ and $\ctxpre^\Gamma$.

\begin{theorem}[CIU is a contextual preorder]\label{thm:ciuctx}
\begin{flalign*}
	&e_1 \ciupre^\Gamma e_2 \implies e_1 \ctxpre^\Gamma e_2
\end{flalign*}
\end{theorem}
\begin{proof}
	This theorem is just a consequence of the compatibility of the logical relations (\Cref{thm:compat}), which coincide with CIU (\Cref{thm:ciulogrel}). Only the proof of transitivity requires simple reasoning in first-order logic.
\end{proof}

\begin{lemma}[Contextual equivalence is closed under substitution]\label{thm:ctxclosed}
\begin{flalign*}
	&e_1 \ctxpre^{\Gamma \cup \{x\}} e_2 \implies \forall v, \valscoped{\Gamma}{v} \implies \subst{e_1}{x \mapsto v} \ctxpre^\Gamma \subst{e_2}{x \mapsto v}
\end{flalign*}
\end{lemma}
\begin{proof}
  This lemma is a consequence of \Cref{ex:beta1} with the expressions $\elet{x}{v}{e_1}$ and $\elet{x}{v}{e_2}$, transitivity, and the fact that CIU equivalence implies contextual equivalence (\Cref{thm:ciuctx}).
\end{proof}

\begin{theorem}[CIU is the greatest contextual preorder]\label{thm:ctxciu}
\begin{flalign*}
	&e_1 \ctxpre^\Gamma e_2 \implies e_1 \ciupre^\Gamma e_2
\end{flalign*}
\end{theorem}
\begin{proof}
  We follow the idea of Wand et al.~\cite{wand2018contextual}. We carry out induction by the size of $\Gamma$.
  \begin{itemize}
  \item If $\Gamma = \emptyset$, both $e_1$ and $e_2$ are closed expressions, that is, we need to prove $e_1 \ciupre e_2$: for any closed frame stack $K$, $\term{K}{e_1} \implies \term{K}{e_2}$. We do induction by the structure of $K$.
  \begin{itemize}
  \item If $K = \idfs$, then we just use the fact of adequacy of $e_1 \ctxpre^\emptyset e_2$ with the empty context to prove $\term{\idfs}{e_1} \implies \term{\idfs}{e_2}$.
  \item If $K = F :: K'$, then we apply \Cref{thm:putback} to the hypothesis, while \Cref{thm:putbackrev} to the goal to be able to apply the induction hypothesis. Now only remains $\subst{F}{e_1} \ctxpre^\emptyset \subst{F}{e_2}$ to prove. After separating cases by the structure of $F$, we can apply the compatibility properties of $\ctxpre^\emptyset$ to finish the proof.
  \end{itemize}

  \item If $\Gamma = \Gamma' \cup \{x\}$, we need to prove that for every $\subscoped{\Gamma' \cup \{x\}}{\sub}{\emptyset}$, $\subst{e_1}{\sub} \ciupre \subst{e_2}{\sub}$. We can also assume, that $x \notin \Gamma'$. We can divide $\sub$ into two parts: $\subst{\subst{e_1}{x \mapsto \sub(x)}}{\sub \restrict \{x\}}$. 
  Now we can apply the induction hypothesis, and the only remaining goal is $\subst{e_1}{x \mapsto \sub(x)} \ctxpre \subst{e_2}{x \mapsto \sub(x)}$ which is proven by \Cref{thm:ctxclosed}.
  \end{itemize}
\end{proof}

Putting \Cref{thm:ciuctx} and \Cref{thm:ctxciu} together, we prove the coincidence of the CIU and contextual equivalence.

\begin{theorem}[CIU theorem]\label{thm:ciuthm}
\begin{flalign*}
	&e_1 \ctxpre^\Gamma e_2 \iff e_1 \ciupre^\Gamma e_2
\end{flalign*}
\end{theorem}

\subsection{Revisiting Behavioural Equivalence}\label{sec:termEquiv}

While defining the logical relations, CIU, and contextual equivalence we used only a termination criterion. But why is termination sufficient for the results of the evaluation to be equivalent? What would it mean, that two expressions, values are equivalent? 
We can take the definition of naive behavioural equivalence, and improve it so that it does not distinguish different function values. For this purpose, we define the equivalence of functions in an application-indexed way, that is equivalent functions should evaluate to the same values after the same number ($n$) of applications, over a limit.

\begin{definition}[Behavioural preorder]
\begin{flalign*}
e_1 \le^R e_2 :=&\ \forall v_1: \rewritesstar{K}{e_1}{v_1} \implies \\
&\ \exists v_2: \rewritesstar{K}{e_2}{v_2} \land (v_1, v_2) \in R \\
v_1 \le^{\textit{val}}_0 v_2 :=&\ \textit{true} \\
l_1 \le^{\textit{val}}_{(1 + n')} l_2 :=&\ l_1 = l_2 \\
\fun{f/k}{x_1}{x_k}{e_1} &\le^{\textit{val}}_{(1 + n')} \fun{f/k}{x_1}{x_k}{e_2} := \\
\forall v_1, \dots, v_k: \valscoped{\emptyset}{v_1} \land \dots \land \valscoped{\emptyset}{v_k} \implies&\\	
\subst{e_1}{f/k \mapsto \fun{f/k}{x_1}{x_k}{e_1}&, x_1 \mapsto v_1, \dots, x_k \mapsto v_k} \le^{\le^{\textit{val}}_{n'}} \\
\subst{e_2}{f/k \mapsto \fun{f/k}{x_1}{x_k}{e_2}&, x_1 \mapsto v_1, \dots, x_k \mapsto v_k} \\
[] \le^{\textit{val}}_{(1 + n')} [] :=&\ \textit{true} \\
[v_1 | v_2] \le^{\textit{val}}_{(1 + n')} [v_1' | v_2'] :=&\ v_1 \le^{\textit{val}}_{n'} v_1' \land v_2 \le^{\textit{val}}_{n'} v_2' \\
\end{flalign*}
\end{definition}

We say, that two values $v_1, v_2$ behave the same way ($v_1 \le^{val} v_2$, note that this is only a preorder relation), when $\forall n:v_1 \le^{val}_{n} v_2 $. Two expressions are equivalent ($e_1 \approx e_2$), if $e_1 \le^{\le^{\textit{val}}} e_2 \land e_2 \le^{\le^{\textit{val}}} e_1 $.

\begin{theorem}[Behavioural equivalence coicides with CIU]\label{thm:behciu}
\begin{flalign*}
	&e_1 \approx e_2 \iff e_1 \ciuequiv e_2
\end{flalign*}
\end{theorem}
\begin{proof}
$\Rightarrow$: Since the inductively defined termination coincides with the reduction-style termination (\Cref{thm:semtermEqTerm}), this direction is just a simple consequence of the definitions.

$\Leftarrow$: Since we have $\rewritesstar{K}{e_1}{v_1}$ for some $v_1$ value, we can show that $\rewritesstar{K}{e_2}{v_2}$ for some $v_2$ by $e_1 \ciuequiv e_2$ and \Cref{thm:semtermEqTerm}. We only need to prove, that $\forall n:v_1 \le^{val}_{n} v_2 $.

We carry out induction on $n$. The case $n = 0$ is $\textit{true}$ by definition. For the case $n = 1 + n'$, we show the induction hypothesis:

$\forall K, v_1, v_2: \rewritesstar{K}{e_1}{v_1} \implies \rewritesstar{K}{e_2}{v_2} \implies v_1 \le^{val}_{n'} v_2$.

Now we do case distinction on $v_1$ and $v_2$.
\begin{itemize}
	\item If the constructors of $v_1$ and $v_2$ differ (e.g. $v_1$ is a literal, $v_2$ is a function, etc.), we construct a contradiction from the hypothesis $e_1 \ciuequiv e_2$. If $v_1$ is not a function, we can use the following idea: if $e_1 \ciuequiv e_2$, then $\term{K'}{e_1} \implies \term{K'}{e_2}$. We choose $K' = K \concat [\case{\Box}{v_1}{0}{\Omega}]$, where we used $\concat$ to denote list concatenation, while $\Omega$ denotes the diverging expression $\applyz{(\funz{f/0}{\applyz{f/0}})}$. 
	
	Since $v_1$ and $v_2$ were constructed differently, and $\term{K'}{e_1}$ holds (for the $K'$ above), therefore, $\term{K'}{e_2}$ should also hold (because $e_1 \ciuequiv e_2$), however this is a divergent configuration, because the pattern matching fails, then $\Omega$ needs to be evaluated, so we got a contradiction.
	
	If $v_1$ is a function, we can use the same idea for the other part of $e_1 \ciuequiv e_2$, i.e. $\term{K'}{e_2} \implies \term{K'}{e_1}$ with $v_2$.
	\item If both expressions were literals, we can use the idea above to prove them equal.
	\item If both expressions were empty lists, then they are equivalent by definition.
	\item If $v_1 = \fun{f/k}{x_1}{x_k}{b_1}$ and $v_2 = \fun{f/k}{x_1}{x_k}{b_2}$, 
	for readability, first we introduce two shorthands $\textit{fun}_1 := \fun{f/k}{x_1}{x_k}{b_1}$ and $\textit{fun}_2 := \fun{f/k}{x_1}{x_k}{b_2}$, and later redefine $v_1$ and $v_2$.
	
	We need to prove, that the bodies of these functions behave the same way when substituting their parameters to equal values.
	That is, in any closed frame stack $K_2$, for any closed values $v_1, \dots v_k, v$, $\rewritesstar{K_2}{\subst{b_1}{f/k \mapsto \textit{fun}_1, x_1 \mapsto v_1, \dots x_k \mapsto v_k}}{v} \implies  \exists v': \rewritesstar{K_2}{\subst{b_2}{f/k \mapsto \textit{fun}_2, x_1 \mapsto v_1, \dots x_k \mapsto v_k}}{v'} \land v \le^{val}_{n'} v'$.
	
	Now we connect the hypotheses, since $\rewritesstar{K}{e_1}{\textit{fun}_1}$ and $\rewritesstar{K_2}{\subst{b_1}{f/k \mapsto \textit{fun}_1, x_1 \mapsto v_1, \dots x_k \mapsto v_k}}{v}$ through \Cref{thm:extendframes} to obtain: $\rewritesstar{K \concat $ $ [\apply{\Box}{v_1}{v_k}] $ $ \concat K_2}{e_1}{v}$.
	
	By $e_1 \ciuequiv e_2$, we also prove, that for some $v'$, $\rewritesstar{K \concat $ $[\apply{\Box}{v_1}{v_k}]$ $ \concat K_2}{e_2}{v'}$. From this hypothesis, after taking some reduction steps (by $\rewritesstar{K}{e_2}{\textit{fun}_2}$ and \Cref{thm:extendframes}), we can prove $\rewritesstar{K_2}{\subst{b_2}{f/k \mapsto \textit{fun}_2, x_1 \mapsto v_1, \dots x_k \mapsto v_k}}{v'}$. We only need to prove, that $v \le^{val}_{n'} v'$, which is done by applying the induction hypothesis, moreover, its premises ($\rewritesstar{K \concat [\apply{\Box}{v_1}{v_k}] \concat K_2}{e_1}{v}$ and $\rewritesstar{K \concat$ $ [\apply{\Box}{v_1}{v_k}]$ $ \concat K_2}{e_2}{v'}$) have already been proved. 
	
	\item If $v_1 = [v_{11} | v_{12}]$ and $v_2 = [v_{21} | v_{22}]$, then we need to show, that $v_{11} \le^{val}_{n'} v_{21}$ and $v_{12} \le^{val}_{n'} v_{22}$. We can do that by applying the induction hypothesis twice, but there are still some evaluations to show (note, that we have $\rewritesstar{K}{e_1}{[v_{11} | v_{21}]}$ and $\rewritesstar{K}{e_2}{[v_{12} | v_{22}]}$):
	\begin{itemize}
		\item For some $K'$, $\rewritesstar{K'}{e_1}{v_{11}} \implies \rewritesstar{K'}{e_2}{v_{12}}$ which can be shown for $K' = K \concat [\case{\Box}{[x | y]}{x}{0}]$ (this is basically a \emph{head} function for lists).
		\item For some $K'$, $\rewritesstar{K'}{e_1}{v_{21}} \implies \rewritesstar{K'}{e_2}{v_{22}}$ which can be shown for $K' = K \concat [\case{\Box}{[x | y]}{y}{0}]$ (this is basically a \emph{tail} function for lists).
	\end{itemize}
\end{itemize}

\end{proof}

\section{Conclusion and Future Work}
\label{sec:conclusion}

In this paper, we described our idea of verifying compound refactorings via decomposition to local transformations. To reason about their correctness, we need a suitable program equivalence definition. Initially we investigated and formalised simple behavioural equivalence~\cite{pierce2010software} in Coq, but this turned out not to be expressive enough since it characterised equivalence as structural rather than semantic.

To solve this issue, in this paper we formalised contextual, CIU preorder and equivalence together with logical relations~\cite{culpepper2017contextual,pitts2000operational,pitts2010step,wand2018contextual}. With these equivalences, we are able to prove non structurally-equivalent functions equivalent when they have the same behaviour.  Moreover, we also presented a proof that reasoning about termination is sufficient to characterise equivalence by giving a formal definition of behavioural equivalence which is also proved to coincide with the other definitions. Our definitions and results are formalised in the Coq~\cite{coreerlangmini} proof assistant.

\paragraph{Future work}
In our previous work~\cite{bereczky2020machine,coreerlang}, we formalised a larger subset of Core Erlang including exceptions and side effects. In the short term, we plan to define the equivalence relations discussed in this paper for this larger language. With the introduction of simple side effects, one could argue about how these effects should count towards the equivalence definition. In our earlier \emph{ad hoc} equivalence definitions we used complete and weak equivalence of the results, where complete equivalence required the same side effects to resolve in the same order during evaluation, while weak equivalence allowed this order to be different for the two expression evaluation. It is therefore a future goal for us to investigate such ``weaker'' definitions of equivalences too.

In the medium and longer term, we plan to extend this formalisation with concurrent language features and also formalise Erlang in full in Coq. Our longer-term goals also include the investigation of bisimulation relations for program equivalence, covering \emph{inter alia} formalised concurrent language features.

\section*{Acknowledgements}





The project has been supported by ÚNKP-20-4 New National Excellence Program of the Ministry for Innovation and Technology and ``Application Domain Specific Highly Reliable IT Solutions'' financed under the Thematic Excellence Programme TKP2020-NKA-06 (National Challenges Subprogramme) funding scheme by the National Research, Development and Innovation Fund of Hungary and ``Integrated program for training new generation of researchers in the disciplinary fields of computer science'', No.  EFOP-3.6.3-VEKOP-16-2017-00002 by the European Union and by the European Social Fund.

\bibliography{mybibfile}

\end{document}